\documentclass[11pt, a4paper]{article}%
\usepackage{amssymb}
\usepackage{amsmath}
\usepackage{graphicx}
\usepackage{color}
\usepackage{xypic}
\usepackage{hyperref}

\usepackage{mathrsfs}

\newcommand\A{\mathbb A}

\newcommand\N{\mathbb N}

\newcommand\Q{\mathbb Q}

\newcommand\Z{\mathbb Z}

\newcommand\CO{{\mathcal O}}

\newcommand\CU{{\mathcal U}}

\newcommand\scrM{{\mathscr M}}

\newtheorem{theorem}{Theorem}
\newtheorem{lemma}[theorem]{Lemma}
\newtheorem{proposition}[theorem]{Proposition}
\newtheorem{remark}[theorem]{Remark}
\newtheorem{corollary}[theorem]{Corollary}
\newtheorem{definition}[theorem]{Definition}

\newtheorem{example}[theorem]{Example}
\newenvironment{proof}[1][Proof]{\noindent\textbf{#1.} }{\
\rule{0.5em}{0.5em}}

\begin{document}

\title{On B\'ezout Inequalities for non-homogeneous Polynomial Ideals\thanks{Partially supported by the following Iranian, Argentinean and Spanish grants: IPM No.95550420 (A.H), UBACyT 20020130100433BA and PICT-2014-3260 (J.H), MTM2014-55262-P (L.M.P).}}

\author{Amir Hashemi$^{({a})}$, Joos Heintz$^{({b})}$,
  Luis Miguel Pardo$^{({c})}$, Pablo Solern\'o$^{({d})}$ \\[0.3cm]
{\text {\small $(a)$ Department of Mathematical Sciences, Isfahan University of Technology,}}\\{\small Isfahan, 84156-83111, Iran}\\{\small
School of Mathematics, Institute for Research in Fundamental Sciences (IPM),}\\{\small Tehran, 19395-5746, Iran} \\[2mm]
{\small $(b)$ Departamento de Computaci\'on and ICC, UBA-CONICET,}\\{\small Facultad de Ciencias Exactas y Naturales, Universidad de Buenos Aires, }\\{\small Ciudad Universitaria, 1428, Buenos Aires, Argentina}\\[2mm]
{\small $(c)$ Depto. de Matem\'aticas, Estad\'istica y Computaci\'on
Facultad de Ciencias,}\\ {\small Universidad de Cantabria,
Avda. Los Castros s/n
E-39071 Santander, Spain}\\[2mm]
{\small $(d)$ Departamento de Matem\'atica and IMAS, UBA-CONICET,}\\ {\small Facultad de Ciencias Exactas y Naturales, Universidad de Buenos Aires,}\\{\small Ciudad Universitaria, 1428, Buenos Aires, Argentina}\\[3mm]
{\small E-mail addresses:}\\ {\small Amir.Hashemi@cc.iut.ac.ir, joos@dc.uba.ar, luis.pardo@unican.es, psolerno@dm.uba.ar}
}

\maketitle

%
%
%
%
%

\begin{abstract}
We introduce a ``workable" notion of degree for non-homogeneous polynomial ideals and formulate and prove ideal theoretic B\'ezout Inequalities for the sum of two ideals in terms of this notion of degree and the degree of generators. We compute probabilistically the degree of an equidimensional ideal.


MSC 13F20, 14A10, 13P10
\end{abstract}

\section*{Introduction}

Motivated by the aim to formulate and prove an idealistic version of B\'ezout's Theorem (see \cite{vogel}) and by applications to transcendence theory (see \cite{MW,Bro}), the notion of degree of \emph{homogeneous} polynomial ideals became intensively studied. In general this work relied on the notion of degree of homogeneous polynomial ideals based on the Hilbert polynomial (see \cite{Lazard} for a different view).

In this paper we propose an alternative and self-contained approach for non-homogeneous polynomial ideals leading to specific results which are not simple consequences of their homogeneous counterparts.

We introduce a ``workable" notion of degree for non-homogeneous polynomial ideals and formulate and prove ideal theoretic B\'ezout Inequalities for the sum of two ideals in terms of this notion of degree. However it turns out, that, due to the presence of embedded primes, a B\'ezout Inequality in completely intrinsic terms (depending only on the degrees of the two given ideals) is unfeasable. Hence in some place the degrees of generators of at least one of the ideals comes into play and our main B\'ezout Inequality  will be of this mixed type.

We finish the paper with a probabilistic algorithm which computes the degree of an equidimensional ideal given by generators.\\

\textbf{Organization of the paper}

The first three sections are devoted to the development of the tools we are going to use in the sequel. The technical highlight is Proposition \ref{prop_bezout1} in the third section which anticipates in some sense the ``correctness" of our ideal theoretic notion of degree (see Theorem \ref{degree_formula} and Definition \ref{degree-primary}) and the main result of the paper, namely the mixed type B\'ezout Inequality, Theorem \ref{bezout_suc_regular}, in the fourth section.

These results become combined in the fifth section with techniques going back to Masser and W\"ustholz \cite{MW} in order to estimate the degree of a polynomial ideal in terms of the degrees of generators (Theorem  \ref{weak}).

Finally the sixth section contains a probabilistic complexity result concerning the computation of the degree of an equidimensional ideal.

\section{Notions and Notations}
Let $K$ be an algebraically closed field, $\vec{X}=(X_1,\ldots,X_n)$ with $X_1,\ldots,X_n$ indeterminates over $K$ and $K[\vec{X}]$ the ring of $n-$variate polynomials with coefficients in $K$. The affine space $K^n$ with the Zariski topology is denoted by $\A^n$. For any ideal ${\frak a}$ of $K[\vec{X}]$, we denote by $V({\frak a})$ the set of its common zeros in $\A^n$ and by $A:=K[\vec{X}]/\frak{a}$ the associated factor ring.

We shall use freely standard notions and notations from Commutative Algebra and affine Algebraic Geometry. These can be found for example in \cite{Atiyah-Macdonald,Matsumura,Shafarevich}.

Less standard is the notion of degree of closed (affine) subvarieties of $\A^n$ we are going to use.

For an \emph{irreducible} closed subvariety $V$ of $\A^n$ we define the degree $\deg V$ of $V$ as the maximum number of points that can arise when we intersect $V$ with an affine linear subspace $E$ of $\A^n$ such that $V\cap E$ is finite (observe that it is a nontrivial fact that $\deg V<\infty$ holds). The degree $\deg V$ of an \emph{arbitrary} closed subvariety $V$ of $\A^n$ is the sum of the degrees of the irreducible components of $V$.

It is a remarkable fact that for this notion of degree that for closed subvarieties $V$ and $W$ an \emph{intrinsic} B\'ezout Inequality holds: $\deg (V\cap W)\le \deg(V)\,.\deg(W)$ (with intrinsic we mean that the degree of $V\cap W$ is estimated in terms of the degrees of $V$ and $W$ only). For more details we refer to \cite{Heintz83,fulton,vogel}.\\

\section{Secant and regular sequences}\label{secant-equences:sec}

Let ${\frak a}\subset K[\vec{X}]$ be an arbitrary ideal of dimension $m$ (i.e. the Krull dimension of the ring $A$ equals $m$) such that all its isolated primes have dimension $m$ (or, equivalently, the variety $V({\frak a})$ is equidimensional of dimension $m$). Under this condition we say that the ideal $\frak{a}$ is \emph{equidimensional} of dimension $m$.

\begin{definition} Let $\vec{f}:=f_1,\ldots,f_{m'}$, $m'\le m$, be a sequence of polynomials of $K[\vec{X}]$. Then $\vec{f}$ is called a \emph{secant sequence for ${\frak a}$} of length $m'$ if for any index $j$, $1\le j\le m'$, the dimension of the ideal ${\frak a}+(f_1,\ldots,f_j)$ is $m-j$. Moreover, $\vec{f}$ is called a \emph{regular sequence of length $m'$ with respect to $\frak{a}$} if $\frak{a}+(f_1,\ldots,f_{m'})$ is a proper ideal and if for any index $j$, $1\le j\le m'$ the polynomial $f_j$ is not a zero-divisor modulo the ideal $\frak{a}+(f_1,\ldots,f_{j-1})$. If this is the case the residue classes $f_1+{\frak a},\ldots,f_{m'}+{\frak a}$ are said to form a regular sequence in $A$.
\end{definition}

\begin{remark}
In the sequel we shall only consider secant sequences of maximal length $m$.
\end{remark}

Observe that for a secant sequence $\vec{f}:=f_1,\ldots,f_m$ for ${\frak a}$, all the isolated components of the ideal ${\frak a}+(f_1)$ have dimension $m-1$ and $f_2,\ldots,f_m$ is a secant sequence for ${\frak a}+(f_1)$. More generally, for $1\le j< m$ the polynomials  $f_{j+1},\ldots,f_m$ form a secant sequence for ${\frak a}+(f_1,\ldots, f_j)$.\\

Any regular sequence of maximal length $m$ with respect to $\frak{a}$ constitutes a secant sequence for $\frak{a}$, because any member of the regular sequence drops the Krull dimension by one at each step, up to reach dimension $0$.

\begin{proposition}\label{morfismo-secante:prop} Let
$\vec{f}:=f_1,\ldots,f_m$ be a secant sequence for the equidimensional ideal ${\frak a}$ of dimension $m$. Let us consider the regular morphism:
\begin{equation}\label{morfismo-secante:eqn}\begin{array}{lccc}
\vec{f}: & V({\frak a}) & \longrightarrow & \A^m,\qquad\qquad\qquad\qquad\qquad\\
& x & \longmapsto & \vec{f}(x):=(f_1(x), \ldots, f_m(x)).
\end{array}
\end{equation}
Then $\vec{f}(V({\frak a}))$ is Zariski dense in $\A^m$. Said otherwise, $\vec{f}:V(\frak{a})\to \A^m$ is dominant.
\end{proposition}
\begin{proof}
 The proof is an easy consequence of the Theorem of the Dimension of the Fibers (see for instance \cite[\S 6.3 Theorem 7]{Shafarevich}). Since the sequence $\vec{f}$ is secant, the fiber $\vec{f}^{-1}(0)$ is a zero-dimensional algebraic set and, therefore, there exists an irreducible component of $V({\frak a})$ which intersects this fiber. If $\mathfrak{p}$ is the corresponding minimal associated prime of ${\frak a}$, the Theorem of the Dimension of the Fibers applied to the restriction $\vec{f}\mid_{V({\frak p})}:V({\frak p})\longrightarrow \A^m$ implies that $\vec{f}\mid_{V({\frak p})}$ is dominant. Thus $\vec{f}(V({\frak p}))$ is Zariski dense in $\A^m$ and hence also $\vec{f}(V(\frak{a}))$.
\end{proof}\\

From Proposition \ref{morfismo-secante:prop} we deduce immediately the following statement:

\begin{corollary} \label{algebraically-independent:corol}
With the same notations as in Proposition \ref{morfismo-secante:prop}, the following is a monomorphism of $K-$algebras:
$$\begin{array}{lcccr}
\vec{f}^*:& K[Y_1,\ldots, Y_m]& \longrightarrow & A=K[\vec{X}]/{\frak a}, & \\
& Y_i & \longmapsto & f_i + {\frak a}, & 1\leq i \leq m.
\end{array}$$
In particular, the residual classes $\{f_1 +{\frak a}, \ldots, f_m+ {\frak a}\}$ are algebraically independent over $K$.
\end{corollary}

\begin{corollary} \label{density:corol}
Let $\vec{f}$ be a secant sequence for the equidimensional ideal $\frak{a}$ of dimension $m$. Then there exists a non-empty Zariski open subset $U$ of $\A^m$, such that for all $\vec{a}:=(\alpha_1,\ldots,\alpha_m)\in U\cap K^m$ the sequence $\vec{f}-\vec{a}:= f_1-\alpha_1,\ldots,f_m-\alpha_m$ is a secant sequence for ${\frak a}$.
\end{corollary}
\begin{proof}
Fix an index $j,\ 1\le j\le m$ and let $\vec{f}_j:V({\frak a})\to \A^j$ be the polynomial map $(f_1(x),\ldots,f_j(x))$. As in the proof of Proposition \ref{morfismo-secante:prop}, we deduce that there exists at least one  irreducible component of $V({\frak a})$ such that the restriction of $\vec{f}_j$ is dominant and its typical fiber has dimension $m-j$. On the other hand, if $\frak{p}$ is an isolated associated prime of $\frak{a}$ such that the typical fiber of the restriction of $\vec{f}_j$ to $V({\frak p})$ has not dimension $m-j$, then the Zariski  closure $\overline{\vec{f}_j(V({\frak p}))}$ is properly contained in $\A^j$. Therefore, the set of points $\vec{y}\in \A^j$ such that the fiber $\vec{f}_j^{-1}(\vec{y})$ is $m-j$ dimensional contains a nonempty Zariski open set $U_j\subset \A^j$.
One verifies immediately that $U:=\displaystyle{\bigcap_{j=1}^m (U_j\times \A^{m-j})\subset \A^m}$ satisfies the statement of the corollary.
\end{proof}

\begin{remark} \label{remark_joos}
There exists in $(\A^n)^m$ a nonempty Zariski open set $\mathcal{O}$ of linear forms of $K[\vec{X}]$ such that each element $\vec{f}=f_1,\ldots,f_m$ of $\mathcal{O}$ constitutes a secant family for $\frak{a}$.
\end{remark}

\begin{proof}
The statement is a consequence of Noether's Normalization Lemma as in \cite[Lemma 1]{Heintz83} applied to the equidimensional variety $V(\frak{a})$.
\end{proof}\\

We need in the sequel the following technical lemma concerning Zariski dense subsets:

\begin{lemma} \label{density:lemmaux}
Let be given positive integers $n_1,\ldots,n_s$ and a subset $U$ of $ \A^{n_1}\times \cdots \times \A^{n_s}$. For each $i=1,\ldots,s$ denote by $\pi_i$ the canonical projection of the product space onto $\A^{n_i}$.
Assume that the set $U$ satisfies the following conditions:
\begin{enumerate}
\item[i)] $\pi_1(U)$ is Zariski dense in $\A^{n_1}$.
\item[ii)] For each $i=1,\ldots,s-1$ and each $a:=(a_1,\ldots,a_i)\in \big(\pi_1\times \ldots\times \pi_i\big)(U)$ the set
\[
\pi_{i+1}\Big((\{a\}\times  \A^{n_{i+1}}\times \cdots \times \A^{n_s})\cap U\Big)\subset \A^{n_{i+1}}
\]
is Zariski dense in $\A^{n_{i+1}}$.
\end{enumerate}
Then, the set $U$ is Zariski dense in $\A^{n_1}\times \cdots \times \A^{n_s}$.
\end{lemma}

\begin{proof}
By induction on $s$. For $s=1$ there is nothing to prove because condition $i)$ implies already the conclusion (the second condition is vacuous). Suppose now $s>1$. For each $i=1,\ldots,s$ denote by $T_i$ the $n_i-$tuple of coordinates of $\A^{n_i}$ and suppose that a polynomial $F(T_1,\ldots,T_s)$ vanishes on $U$. Consider $F$ as polynomial in the variables $T_s$:
\[
F=\sum_{\alpha\in\mathbb{N}_0^{n_s}} f_\alpha(T_1,\ldots,T_{s-1})\ T_s^\alpha.
\]
For an arbitrary point $a:=(a_1,\ldots,a_{s-1})\in \big(\pi_1\times \ldots\times \pi_{s-1}\big)(U)$ the $n_s-$variate polynomial $F(a,T_s)$ vanishes at any point $b\in \A^{n_s}$ with $(a,b)\in U$. In other words, the polynomial $F(a,T_s)$ vanishes on the set $\pi_{s}\big((\{a\}\times \A^{n_s})\cap U\big)$, which is Zariski dense in $\A^{n_s}$ by condition $ii)$ for $i=s-1$. Hence the coefficients $f_\alpha(a)$ are zero for all subindexes $\alpha$ and for all $a\in \big(\pi_1\times \ldots\times \pi_{s-1}\big)(U)$.

We consider now the set $U'\subset \A^{n_1}\times \cdots\times \A^{n_{s-1}}$ defined as
\[
U':=\big(\pi_1\times \ldots\times \pi_{s-1}\big)(U).
\]
It is easy to see that $U'$ satisfies conditions $i)$ and $ii)$. Hence, by induction hypothesis, the set $U'$ is Zariski dense and therefore each polynomial $f_\alpha$ must be identically zero. Hence $F=0$, which implies that $U$ is Zariski dense in $\A^{n_1}\times \cdots \times \A^{n_s}$.
\end{proof}

\begin{corollary} \label{secant_density}
There exists in $(\A^{n+1})^m$ a nonempty Zariski open set $\mathcal O$ of polynomials of degree one such that each element $\vec{f}=f_1,\ldots,f_m$ of $\mathcal O$ constitutes a secant sequence for $\frak{a}$.
\end{corollary}

\begin{proof}
The condition to form a secant sequence is expressible in first order logic and hence constructible. From Corollary \ref{density:corol}, Remark \ref{remark_joos} and Lemma \ref{density:lemmaux} we deduce that the sequences $\vec{f}=f_1,\ldots,f_m$ of degree one polynomials which are secant sequences for $\frak{a}$ constitutes a Zariski dense subset of  $(\A^{n+1})^m$. Since this set is also constructible the corollary follows.
\end{proof}

We are now going to analyze the ubiquity of \emph{regular sequences} with respect to an equidimensional polynomial ideal $\frak a\subset K[\vec{X}]$ of positive dimension. We start with the following simple observation.

\begin{remark} \label{density:one}
Let $\frak{p}$ be a prime ideal of the polynomial ring $K[\vec{X}]$ of positive dimension. Then the set
\[
T:=\{(a_1,\ldots,a_n)\in K^n\ |\ \textit{there\ exists}\ \lambda\in K \ \textit{such\ that}\  a_1X_1+\cdots+a_nX_n+\lambda \in \mathfrak{p}\}
\]
is a proper linear subspace of $K^n$.
\end{remark}

\begin{proof}
The set $T$ is clearly a linear subspace because $\mathfrak{p}$ is an ideal. If $T=K^n$ the elements of the canonical basis of $K^n$ belong to $T$. Hence, there exist scalars $\lambda_1,\ldots,\lambda_n$ such that $X_1+\lambda_1,\ldots,X_n+\lambda_n$ belong to $\mathfrak{p}$. Thus the ideal $\frak{p}$ is maximal. This contradicts the assumption that $\frak{p}$ is of positive dimension.
\end{proof}\\

From Remark \ref{density:one} and Lemma \ref{density:lemmaux} we deduce now the following result about the density of degree one regular sequences with respect to an equidimensional polynomial ideal.

\begin{proposition} \label{density:corol2}
Let $\frak{a}\subset K[\vec{X}]$ be an equidimensional ideal of dimension $m>0$. Then there exists a Zariski dense subset $U$ of $(\A^{n+1})^m$ such that for all $(\vec{a_1},\ldots,\vec{a_m})\in U$ with $\vec{a_i}:=(a_{1i},\ldots,a_{ni},a_{(n+1)i})$, $1\le i\le m$, the polynomials $\ell_1,\ldots,\ell_m$ of degree $1$ defined by
\[\ell_i:= a_{1i}X_1+\cdots+a_{ni}X_n+a_{(n+1)i}\]
form a regular sequence with respect to the ideal $\frak{a}$.
\end{proposition}

\begin{proof}
We start by the construction of a suitable regular generic polynomial of degree $1$. Let $\frak{p}_1,\ldots,\frak{p}_t$ be the associated primes of $\mathfrak{a}$ of positive dimension and let  $\frak{p}_{t+1},\ldots,\frak{p}_r$ those  associated primes which are maximal ideals.  Observe $1\le t$ since the ideal $\frak{a}$ is of positive dimension. For each $j=1,\ldots,t$ let $T_j$ be the proper linear subspace of $\A^n$ associated to $\mathfrak{p}_j$ following Remark \ref{density:one}. Thus $U_1:=\A^n\setminus\bigcup_{j=1}^t T_j$ is a constructible Zariski dense subset of $\A^n$. Observe that for $1\le j\le t$ and any (homogeneous) linear form $\ell$ whose coefficients belong to $U_1$, the constructible set $\ell(V(\mathfrak{p}_j))$ is Zariski dense in $\mathbb{A}^1$. Thus, the intersection $U_\ell:=\bigcap_{j=1}^s-\ell(V(\mathfrak{p}_j))$ is constructible and Zariski dense too. Hence $\mathcal{U}:=\{(\ell,u)\ ;\ \ell\in U_1,u\in U_\ell\}$ is a constructible subset of $\A^{n+1}$ which is Zariski dense following Lemma \ref{density:lemmaux}.

Now, for each maximal ideal $\mathfrak{p}_j$, $t< j \le r$, associated to $\mathfrak{a}$ we consider $W_j\subset \A^{n+1}$, the $n-$dimensional linear subspace of the polynomials of degree one contained in $\frak{p}_j$. Since $\mathcal{U}$ is Zariski dense and constructible in $\A^{n+1}$, we conclude that
\[
\mathcal{U}_1:=\mathcal{U}\setminus \bigcup_{j=t+1}^r W_j
\]
is constructible and Zariski dense in $\A^{n+1}$.

Now we consider an arbitrary polynomial $\ell_1:=a_1X_1+\cdots a_nX_n+a_{n+1}$ with $(a_1,\ldots,a_{n+1})\in \mathcal{U}_1$. Then, $\ell_1\notin \frak{p}_j$ for $t<j\le r$  since the vector $(a_1,\ldots,a_{n+1})$ does not belong to $\cup_{j=t+1}^r W_j$. On the other hand, we have $\ell_1\notin \frak{p}_j$ for all $1\le j\le t$ because the homogeneous part of $\ell_1$ is in $U_1$.

In other words, for $\ell_1\in \mathcal{U}_1$ we infer that $\ell_1$  does not belong to any associated prime of $\frak{a}$. Thus $\ell_1$ is not a zero divisor mod $\frak{a}$.

In order to prove that $\ell_1$ is a regular element it suffices to show that $\ell_1$ is not a unity modulo $\frak{a}$. Otherwise $\ell_1$ must be also a unity mod $\frak{p}_1$, i.e. $\ell_1(p)\ne 0$ for all $p\in V(\mathfrak{p}_1)$ and then $-a_{n+1}$ is not in the image $\ell(V(\mathfrak{p}_1))$ which contradicts the fact that $a_{n+1}\in U_\ell$ (recall that $1\le t$ holds).

Summarizing, we have shown that there exists a (constructible) Zariski dense subset $\mathcal{U}_1$ of $\A^{n+1}$ such that any polynomial of degree one with coefficients in $\mathcal{U}_1$ is regular mod $\frak{a}$.

Now the corollary follows by an inductive argument based on Lemma \ref{density:lemmaux}: take any $\ell_1$ with coefficients in $\mathcal{U}_1$. Then by Krull's Principal Ideal Theorem, the ideal $\mathfrak{a}_1:=\mathfrak{a} +(\ell_1)$ is equidimensional of dimension $m-1$. If $m=1$ there is nothing to prove.

Assume $m>1$. We build a constructible Zariski dense subset $\mathcal{U}_{2,\ell_1}$ of $\A^{n+1}$ (depending on $\ell_1$) such that for any $\ell_2$ with coefficients in this set, the pair $\ell_1,\ell_2$ is a regular sequence of length $2$ for $\mathfrak{a}$. By repeating this argument for each $\ell_1$, Lemma \ref{density:lemmaux} ensures the existence of a Zariski dense subset $\mathcal{U}_2$ of $\A^{n+1}\times \A^{n+1}$ representing the coefficients of regular sequences of length $2$ for $\mathfrak{a}$, formed by polynomials of degree $1$. The corollary follows now inductively after $m$ steps.
\end{proof}

\section{Generic fiber vs. special fibers}

Throughout this section let notations and assumptions be the same as in the previous section. Let us now consider a minimal primary decomposition of the equidimensional ideal
${\frak a}$, of dimension $m$, namely
$${\frak a}={\frak q}_1\cap \cdots \cap {\frak q}_r,$$
where each ideal ${\frak q}_i$ is ${\frak p}_i-$primary. Assume $\vec{f}$ is a secant sequence with respect to the ideal ${\frak a}$. Let us denote by $K[\vec{f}]$ the image of $K[Y_1,\ldots, Y_m]$ in $A=K[X_1,\ldots, X_n]/{\frak a}$.

Let us recall the polynomial mapping
$\vec{f}: V({\frak a})\longrightarrow \A^m$ and let $\vec{f}_i$ be its restriction to each of the components $V({\frak p}_i)$:
$$\vec{f}_i:=\vec{f}\mid_{V({\frak p}_i)}: V({\frak p}_i)\longrightarrow \A^m,\textrm{\ where\ } 1\le i\le r.$$

Taking into account Proposition \ref{morfismo-secante:prop} we may assume without loss of generality that there exists an index $1\le s\le r$ such that $\vec{f}_i$ is dominant if and only if $1\le i\le s$.
\begin{lemma}\label{primarias-jugando:lema} With these notations the following statements hold:
\begin{enumerate}
\item[i)] For any $i$, $s+1\leq i \leq r$, there exists a non-zero polynomial $Q_i\in K[Y_1,\ldots, Y_m]$ such that $Q_i(f_1,\ldots, f_m)$ belongs  to ${\frak q}_i$.
\item[ii)] For any $i$, $1\leq i \leq s$, and for any $j$, $1\leq j \leq n$, there is a polynomial
$Q_{ij}\in K[Y_1,\ldots, Y_m][T]$, of positive degree in the variable $T$ such that $Q_{ij}(f_1,\ldots, f_m)(X_j)$ belongs to ${\frak q}_i$.
\end{enumerate}
\end{lemma}
\begin{proof} For the proof of statement $i)$, observe that for any $i$, $s+1\leq i \leq r$, the image $\vec{f}(V({\frak p}_i))$ is contained in a proper hypersurface of $\A^m$. Taking $G_i\in K[Y_1,\ldots, Y_m]$ the minimal equation of this hypersurface, we conclude that $G_i(f_1,\ldots, f_m)$ belongs to ${\frak p}_i$. Since the radical of ${ {\frak q}_i}$ is ${\frak p}_i$, there exists some power $Q_i:=G_i^{k_i}$ such that $Q_i(f_1,\ldots, f_m)$ belongs to ${\frak q}_i$, as wanted.

For the proof of statement $ii)$, let us write by ${\frak p}:={\frak p}_i$ where $1\leq i \leq s$. As $\vec{f}_i(V({\frak p}))$ is Zariski dense in $\A^m$, and $\frak{a}$ is $m-$equidimensional we have a monomorphism of finitely generated $K-$algebras of the same Krull dimension:
    $$\vec{f}_i^* : K[Y_1,\ldots, Y_m]\hookrightarrow K[X_1,\ldots, X_n]/{\frak p},$$
    with $\vec{f}_i^*(Y_k) = f_k+ {\frak p}$, $1\le k\le m$. Thus we may consider $K[\vec{f}]$ as a subalgebra of $K[X_1,\ldots,X_n]/\frak{p}$. Let $S:=K[\vec{f}]\setminus \{0\}$, Then $K(\vec{f}):=S^{-1}K[\vec{f}]$ is the quotient field of $K[\vec{f}]$ and we have an extension of domains of the same Krull dimension:
    \[K(\vec{f})\hookrightarrow S^{-1}(K[X_1,\ldots,X_n]/{\frak p}),\]
    which is therefore a finite field extension. Then, for any $1\le j\le n$ the residue class  $X_j+{\frak p}$ is algebraic over $K(\vec{f})$. In particular, there is a polynomial $H_j\in K[Y_1,\ldots, Y_m][T]$ (depending on $\frak{p}$), of positive degree in $T$ such that $H_j(f_1,\ldots, f_m)(X_j)$ belongs to ${\frak p}$. As ${\frak p}= {\frak p}_i$, let us write $G_{ij}:=H_j$. Since $\frak{p}_i$ is the radical of $\frak{q}_i$ , there is some power $Q_{ij}:=G_{ij}^{k_{ij}}$ such that $Q_{ij}(f_1,\ldots, f_m)(X_j)$ belongs to ${\frak q}_i$ as wanted.
\end{proof}\\

For the next statement observe that Corollary \ref{algebraically-independent:corol} implies that for any nonzero polynomial $q\in K[Y_1,\ldots,Y_m]$ we have $q(\vec{f})=q(f_1,\ldots,f_m)\ne 0$. Thus $K[\vec{f}]\setminus \{0\}$ forms a multiplicative closed set of $A$ which we denote in the sequel by $S(\vec{f})$. Observe that $K(\vec{f}):=S(\vec{f})^{-1}K[\vec{f}]$ is the quotient field of $K[\vec{f}]$.

\begin{proposition} \label{equality1:prop}
There is a nonzero polynomial $q\in K[Y_1,\ldots, Y_m]$ such that for $q(\vec{f})=q(f_1,\ldots,f_m)\in S(\vec{f})$ the localizations of $K[\vec{f}]$ and $A$ by $q(\vec{f})$
define an integral ring extension:
\[K[\vec{f}]_{q(\vec{f})}\hookrightarrow A_{q(\vec{f})}.\]
In particular $A_{q(\vec{f})}$ is a finite $K[\vec{f}]_{q(\vec{f})}-$module.
\end{proposition}

\begin{proof}
We use the notations of Lemma \ref{primarias-jugando:lema}. Let us consider the polynomial $q\in K[Y_1,\ldots, Y_m]$ given by
\[q:=\left(\prod_{i=1}^s \prod_{j=1}^n a_{ij}(Y_1,\ldots, Y_m)\right)\times \left( \prod_{i=s+1}^r Q_i(Y_1,\ldots, Y_m)\right)\in K[Y_1,\ldots, Y_m],\]
where for $1\le i\le s$ and $1\le j\le n$ the polynomial $a_{ij}$ is the non-zero leading coefficient of the polynomial $Q_{ij}\in K[Y_1,\ldots, Y_m][T]$ which satisfies the condition $a_{ij}(f_1,\ldots, f_m)\not\in {\frak p}_i$ (see the proof of Lemma \ref{primarias-jugando:lema}). As $q$ is a non-zero polynomial we have $q(\vec{f})=q(f_1,\ldots,f_m)\not\in {\frak a}$. Thus we obtain a $K-$algebra extension
$$K[\vec{f}]_{q(\vec{f})}\hookrightarrow A_{q(\vec{f})}.$$
As $q(\vec{f})\in {\frak q}_i$ for $s+1\leq i \leq r$, the minimal primary decomposition of the ideal $(0)$ in
$A_{q(\vec{f})}$ is given by:
$$(0)= {\frak q}_1^e \cap \cdots \cap {\frak q}_s^e,$$
where ${\frak q}_i^e$ is the extension of ${\frak q}_i$ to $A_{q(\vec{f})}$. In particular, the ring extension above
is integral since for every $1\le j\le n$ the residue class $X_j+\frak{a}$ satisfies the algebraic dependence equation given by $\prod_{i=1}^s Q_{ij}$ described in Lemma \ref{primarias-jugando:lema}. In particular, as $A_{q(\vec{f})}$ is a finitely generated $K[\vec{f}]_{q(\vec{f})}-$algebra and as every residue class $X_j+\frak{a}$, $1\le j\le n$ is integral over $K[\vec{f}]_{q(\vec{f})}$, the algebra
$A_{q(\vec{f})}$ is a finite $K[\vec{f}]_{q(\vec{f})}-$module.\end{proof}\\

We deduce from Proposition \ref{equality1:prop} the following statement.

\begin{corollary} \label{finite} The localized ring $S(\vec{f})^{-1}A$ is a $K(\vec{f})-$algebra of finite dimension.
\end{corollary}

\begin{proposition}\label{libre-multiplicativo:prop}
There is a nonzero polynomial $p\in K[Y_1,\ldots, Y_m]$ such that $p(f):=p(f_1,\ldots,f_m)\in S(\vec{f})$ has the following property. The localizations by $p(f)$
define an integral ring extension:
$$K[\vec{f}]_{p(\vec{f})}\hookrightarrow A_{p(\vec{f})},$$
and $A_{p(\vec{f})}$ is a free $K[\vec{f}]_{p(\vec{f})}-$module of finite rank. Moreover,  its rank satisfies the condition
$$\textrm{rank}_{K[\vec{f}]_{p(\vec{f})}}(A_{p(\vec{f})})= \dim_{K(\vec{f})}S(\vec{f})^{-1}A.$$
\end{proposition}
\begin{proof} Let $q\in K[Y_1,\ldots, Y_m]$ be the non-zero polynomial of Proposition \ref{equality1:prop}. Then, we have
\begin{itemize}
\item a finite subset $\beta$ of $A$, such that $\beta$ is a basis of $S(\vec{f})^{-1}A$ as $K(\vec{f})-$vector space.
\item a finite subset $\scrM$ of $A$, such that $\scrM$ is a system of generators of the
$K[\vec{f}]_{q(\vec{f})}-$module $A_{q(\vec{f})}$.
\end{itemize}
As $\beta$ is a basis of $S(\vec{f})^{-1}A$ as $K(\vec{f})-$vector space, there is some nonzero polynomial $h\in K[Y_1,\ldots,Y_m]$ such that $h(\vec{f}):=h(f_1,\ldots, f_m)\not=0$ in $A$ and such that all elements in
$\scrM$ are linear combinations of the elements in the basis $\beta$ with coefficients in $K[\vec{f}]_{h(\vec{f})}$. Then, $p:=qh\in K[Y_1,\ldots, Y_m]$ is a non-zero polynomial such that
$$K[\vec{f}]_{p(\vec{f})}\hookrightarrow A_{p(\vec{f})}$$
is an integral ring extension and such that $A_{p(\vec{f})}$ is a free $K[\vec{f}]_{p(\vec{f})}-$module of finite rank. As rank is stable under localizations, we conclude
$$\textrm{rank}_{K[\vec{f}]_{p(\vec{f})}}(A_{p(\vec{f})})= \dim_{K(\vec{f})}S^{-1}A.$$
\end{proof}

\begin{lemma} \label{inequality}
Suppose that $\vec{f}$ is a regular sequence. Then the inequality
$\dim_K A/(\vec{f})\le \dim_{K(\vec{f})} S(\vec{f})^{-1}A
$ holds.
\end{lemma}

\begin{proof}
By induction on the Krull dimension $m$ of $A$.

If $m=0$ there is nothing to prove. Suppose $m\ge 1$ and denote by $B:=A/(f_1)$. Since $\vec{f}$ is a secant sequence, the ring $B$ has Krull dimension $m-1$ and $\vec{f'}:=f_2,\ldots,f_m$ is a regular sequence for the ideal ${\frak a}+(f_1)$. Therefore, if $T$ denotes the multiplicative subset $K[\vec{f'}]\setminus \{0\}$ of $B$ and $L':=K(\vec{f'})$ we have by induction hypothesis the inequality $\dim_K B/(\vec{f'})\le \dim_{L'} T^{-1}B$.

Since $B/(\vec{f'})=A/(\vec{f})$ it suffices to show the inequality $\dim_{L'} T^{-1}B\le \dim_{K(\vec{f})} S^{-1}A$.

Let $g_1,\ldots,g_s$ be polynomials in $K[\vec{X}]$ whose classes form a basis of $T^{-1}B$ over $L'$. We consider the images of $g_1,\ldots,g_s$ in the ring $S(\vec{f})^{-1}A$. We show that these elements are linearly independent over $K(\vec{f})$. Suppose on the contrary that they are $K(\vec{f})-$linearly dependent. Cleaning denominators we may suppose that there exists a non trivial linear combination $\sum_{i=1}^s p_i g_i$ which belongs to the ideal ${\frak a}$ where the $p_i$'s are polynomials in $K[\vec{f}]$. Since the polynomials $f_1,\ldots,f_m$ are algebraically independent modulo $\frak{a}$, $f_1\in K[\vec{f}]$ may be viewed as an irreducible element in a factorial domain. Moreover, since $f_1$ is not zero-divisor modulo ${\frak a}$ we may suppose without loss of generality that $f_1$ is not a common factor of all polynomials $p_1,\ldots,p_s$. Taking the class of each polynomial $p_i$ modulo the ideal $(f_1)\subset K[\vec{f}]$ we conclude that the polynomials $g_i$ are not linearly independent in $T^{-1}B$, contradiction.
\end{proof}\\

Lemma \ref{inequality} fails for secant sequences:

\begin{example} \label{cex}
{\rm Consider the one dimensional ideal ${\frak a}=(X_1^3,X_1^2X_2)\subset K[X_1,X_2]$ whose primary decomposition is ${\frak a}=(X_1^2)\cap (X_1^3,X_2)$ . The polynomial $X_2$ is a secant sequence for ${\frak a}$ and $A/(X_2)=K[X_1,X_2]/(X_1^3,X_2)\simeq K[X_1]/(X_1^3)$. Therefore $\dim_K A/(X_2)=3$.}

\noindent \rm{On the other hand, if $S:=K[X_2]\setminus \{0\}$ and $L:=K(X_2)$, we have \[S^{-1}A\simeq S^{-1}K[X_1,X_2]/(X_1^3,X_1^2)\simeq K(X_2)[X_1]/(X_1^2)\] (since $X_2$ is invertible in $S^{-1}A$). But then $\dim_{L}S^{-1}A=2$.}

\noindent \rm{Observe that $X_2$ is a secant family for ${\frak a}$ but it is a zero divisor in $A$ since $X_2 X_1^2\in {\frak a}$ and $X_1^2\notin {\frak a}$.}
\end{example}

Typically we have equality in Lemma \ref{inequality}.

\begin{proposition} \label{equality}
Let $\vec{f}$ be a secant sequence for $\frak{a}$, as before. Then there exists a non-empty Zariski open subset $\mathcal O$ of $\A^m$ such that for any $\vec{a}\in\mathcal O$ the equality $\dim_K A/(\vec{f}-\vec{a})=\dim_{K(\vec{f})}S(\vec{f})^{-1}A$  holds.
\end{proposition}

\begin{proof}
By  Proposition \ref{libre-multiplicativo:prop}, there exists a non-zero polynomial $p\in K[Y_1,\ldots, Y_m]$ such that
the localizations by $p(\vec{f})$
define an integral ring extension
$$K[\vec{f}]_{p(\vec{f})}\hookrightarrow A_{p(\vec{f})}$$
and $A_{p(\vec{f})}$ is a free $K[\vec{f}]_{p(\vec{f})}-$module of finite rank
$$N:=\textrm{rank}_{K[\vec{f}]_{p(\vec{f})}}(A_{p(\vec{f})})= \dim_{K(\vec{f})}S(\vec{f})^{-1}A.$$

Let $H:=\{p\ne 0\}\subset \A^m$ and let $U$ be the nonempty Zariski open subset of $\A^m$ of Corollary \ref{density:corol}. Then, $\mathcal{O}:=H\cap U$ is also a non-empty Zariski open subset of $\A^m$.  Let $\vec{a}\in \mathcal{O}$ be a point in this open set and let us denote by ${\frak m}_{\vec{a}}$ the ideal in $K[\vec{f}]$ generated by
the sequence $f_1-a_1,\ldots, f_m-a_m$ in
$K[\vec{f}]$.  As $K[\vec{f}]$ is a polynomial ring, ${\frak m}_{\vec{a}}$ is a maximal ideal in $K[\vec{f}]$. Let us denote by ${\frak m}_{\vec{a}}A$ its extension to $A$. Moreover, as $p(\vec{a})\not=0$, we have $p(\vec{f})\not\in {\frak m}_{\vec{a}}$. In particular,
the extension ${\frak m}_{\vec{a}}^{(1)}:={\frak m}_{\vec{a}}K[\vec{f}]_{p(\vec{f})}$ is also a maximal ideal in $K[\vec{f}]_{p(\vec{f})}$.
Let us consider the submodule ${\frak m}_{\vec{a}}^{(1)}A_{p(\vec{f})}$ of $A_{p(\vec{f})}$ as  $K[\vec{f}]_{p(\vec{f})}-$module. Note that $(A/{\frak m}_{\vec{a}}A)_{p(\vec{f})}$ is isomorphic
as $K[\vec{f}]_{p(\vec{f})}-$module to
$$\left(K[\vec{f}]_{p(\vec{f})}\right)^N/{\frak m}_{\vec{a}}^{(1)}\left(K[\vec{f}]_{p(\vec{f})}\right)^N.$$

 Thus, we conclude that as
$K[\vec{f}]_{p(\vec{f})}-$module, $(A/{\frak m}_{\vec{a}} A)_{p(\vec{f})}$ is isomorphic to
$$\left(K[\vec{f}]_{p(\vec{f})}/ {\frak m}_{\vec{a}}^{(1)}\right)^N\cong \left(K[\vec{f}]/ {\frak m}_{\vec{a}}\right)^N\cong K^N.$$
Hence $\dim_K A/(\vec{f}-\vec{a})=N=\dim_{K(\vec{f})}S(\vec{f})^{-1}A$ for all $\vec{a}\in \mathcal{O}$.
\end{proof}

\begin{corollary} \label{dense_local}
There exists a Zariski dense subset $\mathcal{U}$ of $(\A^{n+1})^m$ such that for any sequence $\vec{f}=f_1,\ldots,f_m$ of degree one polynomials of $\mathcal{U}$ the sequence $\vec{f}$ is secant and the equality $\dim_{K(\vec{f})}S(\vec{f})^{-1}A=\dim_K A/(\vec{f})$ holds.
\end{corollary}

\begin{proof}
Combine Corollary \ref{secant_density} with Proposition \ref{equality} and Lemma \ref{density:lemmaux}.
\end{proof}\\

Let $1\le q\le r$ be the number of isolated primes of the ideal $\frak{a}$. Without loss of generality we may assume that these are $\frak{p}_1,\ldots,\frak{p}_q$. Then for each $1\le j\le q$ the ring $A_{\frak{p}_j}$ is local and Artinian with maximal ideal $(\frak{p}_j/\frak{a})_{\frak{p}_j}$. In the sequel we denote by $\ell_j$ the length of $A_{\frak{p}_j}$.

As in Lemma \ref{primarias-jugando:lema} we may assume without loss of generality that there exists an index $1\le s\le q$ such that $\vec{f}_i=\vec{f}|_{V(\frak{p}_i)}:V(\frak{p}_i)\to \A^m$ is dominant if and only if $1\le i\le s$.

We say that the ring extension $K(\vec{f})\subseteq S(\vec{f})^{-1}A$ is \emph{separable} if for every $1\le i\le s$ the field extension
$K(\vec{f})\subseteq S(\vec{f})^{-1}A/({\frak p}_i/\frak{a})$ is separable (compare Corollary \ref{finite}).

With these notions and notations we may formulate the following result.

\begin{proposition} \label{prop_bezout1}\
\begin{enumerate}
\item[i)] Assume that the ring extension $K(\vec{f})\subseteq S(\vec{f})^{-1}A$ is separable. Then we have
\[\dim_{K(\vec{f})}S(\vec{f})^{-1} A\leq \left(\sum_{i=1}^s \deg V({\frak p}_i)\ \ell_i\right)\prod_{k=1}^m \deg(f_k)\le\]
\[\le\left(\sum_{j=1}^q \deg V({\frak p}_j)\   \ell_j\right)\prod_{k=1}^m \deg(f_k).\]

\item[ii)] There exists a nonempty Zariski open subset $\mathcal{O}$ of $(\A^{n+1})^m$ such that for any sequence $\vec{f}=f_1,\ldots,f_m$ of degree one polynomials of $\mathcal{O}$, the sequence $\vec{f}$ is secant for $\frak{a}$ and such that
\[\dim_{K(\vec{f})}S(\vec{f})^{-1} A=\sum_{j=1}^q \deg V(\frak{p}_j)\ \ell_j\]
holds.
\end{enumerate}
\end{proposition}

\begin{proof}
We are going to show statement $i)$. Let us abbreviate $L:=K(\vec{f})$, $S:=S(\vec{f})$ and $R:=S(\vec{f})^{-1}A=S^{-1}A$.

By virtue of Lemma \ref{primarias-jugando:lema} $i)$ the localizations $\frak{m}_i:=S^{-1}(\frak{p}_i/\frak{a})$, $1\le i\le s$, are exactly the maximal ideals of $R$. By Corollary \ref{finite} the $L-$algebra $R$ is finite dimensional and therefore Artinian. From the Chinese Remainder Theorem we deduce now $\displaystyle{R=\bigoplus_{i=1}^s R_{\frak{m}_i}}$. Let $1\le i\le s$ and observe that the rings $R_{\frak{m}_i}$ and $A_{\frak{p}_i}$ are isomorphic.

Therefore $\ell_i$ is also the length of the Artinian local ring $R_{\frak{m}_i}$. This implies $\dim_L R_{{m}_i}=[R/\frak{m}_i:L]\,\ell_i$.

Putting all this together we obtain
\begin{equation} \label{putting}
\dim_L S^{-1}A=\dim_L R=\sum_{i=1}^s [R/\frak{m}_i:L]\,\ell_i.
\end{equation}

Fix again $1\le i\le s$. We are now going to analyze in geometric terms the quantity  $[R/\frak{m}_i:L]$.

Observe that the dominating morphism $\vec{f}_i=\vec{f}|_{V(\frak{p}_i)}:V(\frak{p}_i)\to \A^m$ induces canonical field isomorphisms
\[
R/\frak{m}_i \cong S^{-1}(A/(\frak{p_i}/\frak{a})) \cong K(V(\frak{p}_i)),
\]
where $K(V(\frak{p}_i))$ denotes the fraction field of $V(\frak{p}_i)$. Since by assumption the field extension $K(\vec{f})\subset S^{-1}(A/(\frak{p_i}/\frak{a}))$ is separable, we see that $\vec{f}_i$ is generically unramified. In particular, we are in conditions to apply \cite[Theorem 4, Chapter II, \S6.3]{Shafarevich} to conclude that there exists a Zariski open subset $W$ of $\A^m$ such that for any point $\vec{a}\in W$ the fiber $\vec{f}_i^{-1}(\vec{a})$ is unramified and of cardinality $[R/\frak{m}_i:L]$.

Choose $\vec{a}\in W$. From the B\'ezout Inequality \cite{Heintz83} we deduce now
\[
[R/\frak{m}_i:L]=\sharp \vec{f}_i^{-1}(\vec{a})\le \deg V(\frak{p}_i)\prod_{k=1}^m \deg(f_k).
\]
From (\ref{putting}) we infer finally the statement $i)$ of the proposition, namely
\[\dim_{K(\vec{f})}S(\vec{f})^{-1} A = \dim_L S^{-1}A=\sum_{i=1}^s \sharp \vec{f}_i^{-1}(\vec{a})\, \ell_i\le
\]
\[
\le \left(\sum_{i=1}^s \deg V({\frak p}_i)\ \ell_i\right)\prod_{k=1}^m \deg(f_k)\le \left(\sum_{j=1}^q \deg V({\frak p}_j)\   \ell_j\right)\prod_{k=1}^m \deg(f_k).
\]

We are now going to prove statement $ii)$ of the proposition.

Following Corollary \ref{secant_density} and \cite[Lemma 1]{Heintz83} we may choose a nonempty Zariski open subset of $(\A^{n+1})^m$ such that for any sequence $\vec{f}=f_1,\ldots,f_m$ of degree one polynomials of $\mathcal{O}$ the following conditions are satisfied:
\begin{itemize}
\item $\vec{f}$ is a secant sequence for $\frak{a}$;
\item for $1\le j\le q$ the morphisms $\vec{f}_j:=\vec{f}\mid_{V({\frak p}_j)}: V({\frak p}_j)\to \A^m$ are dominant, generically unramified and of degree $\deg V(\frak{p}_j)$.
\end{itemize}

Thus, in particular, any sequence $\vec{f}=f_1,\ldots,f_m$ of  $\mathcal{O}$ fulfills the assumption of the statement $i)$ of the proposition.

From the proof of statement $i)$ we deduce:
\[ \dim_{K(\vec{f})}S(\vec{f})^{-1}A=\sum_{j=1}^q \sharp \vec{f}_j^{-1} (\vec{a})\, \ell_j
\]
for a suitable, generically chosen point $\vec{a}\in \A^m$. Since $f_j$ is generically unramified of degree $\deg V(\frak{p}_j)$ we may choose $\vec{a}$ such that $\sharp \vec{f}_j^{-1} (\vec{a})=\deg V(\frak{p}_j)$ holds for $1\le j\le q$. This implies the statement $ii)$ of the proposition, namely
$\dim_{K(\vec{f})}S(\vec{f})^{-1}A= \sum_{j=1}^q \deg V(\frak{p}_j)\, \ell_j.$
\end{proof}

\section{Notion of the degree of an ideal of non-homogeneous polynomials}

As at the end of the last section let be given an equidimensional ideal $\frak{a}$ of $K[\vec{X}]$ of dimension $m$ with isolated primes $\frak{p}_1,\ldots,\frak{p}_q$. Let $A:=K[\vec{X}]/\frak{a}$, recall that for $1\le j\le q$ the ring $A_{\frak{p}_j}$ is local and Artinian and let $\ell_j$ be the length of $A_{\frak{p}_j}$.

We shall need the following technical result.

\begin{lemma} \label{density_dimension}
There exists a nonempty Zariski open subset $\mathcal{U}$ of $(\A^{n+1})^m$ such that for any sequence $\vec{f}=f_1,\ldots,f_m$ of degree one polynomials of $\CU$, the sequence $\vec{f}$ is secant for $\frak{a}$ and $\dim_K A/(\vec{f})$ is constant, independently from $\vec{f}$.
\end{lemma}

\begin{proof}
Let $T_{ij}$, $1\le i\le m$, $0\le j\le n$ be new indeterminates, $\vec{T}=(T_{ij})_{{1\le i\le m}\atop{0\le j\le n}}$ and let $F_i:=\sum_{j=1}^m T_{ij}X_j\, +\, T_{i0}$, $1\le i\le m$. Fix any monomial order of $\vec{X}$ and let $g_1,\ldots g_s\in K[\vec{X}]$ be a set of generators of $\frak{a}$. Observe that the ideal $(g_1,\ldots g_s,F_1,\ldots,F_m)$ of $K(\vec{T})[\vec{X}]$ is zero-dimensional.

Consider an arbitrary Gr\"obner basis computation $\beta$ of this ideal.

The leading coefficients occurring in $\beta$ form a finite set of nonzero rational functions of $K(\vec{T})$. Hence, there exists a nonempty Zariski open subset $\CU$ of $\A^{(n+1)m}$ where none of the numerators and denominators of these rational functions vanishes.

Let $\vec{f}=f_1,\ldots,f_m$ be a sequence of degree one polynomials of $\CU$. Then $\beta$ may be specialized to a Gr\"obner basis computation of $(g_1,\ldots g_s,f_1,\ldots,f_m)$ in $K[\vec{X}]$ which yields the stair of $\beta$.

In view of Corollary \ref{secant_density} we may assume without loss of generality that for every sequence $\vec{f}=f_1,\ldots,f_m$ of degree one polynomials of $\CU$ the sequence $\vec{f}$ is secant for $\frak{a}$.

Therefore everything is well defined and $\dim_K A/(\vec{f})$ is finite and constant on $\CU$.
\end{proof}

\begin{theorem} \label{degree_formula}
There exists a nonempty Zariski open subset $\CO$ of $(\A^{n+1})^m$ such that for any sequence $\vec{f}=f_1,\ldots,f_m$ of degree one polynomials of $\CO$ the sequence $\vec{f}$ is secant for $\frak{a}$ and $\dim_K A/(\vec{f})=\sum_{j=1}^q \deg V(\frak{p}_j)\, \ell_j$ holds.
\end{theorem}

\begin{proof}
Combining Proposition \ref{prop_bezout1} $ii)$ with Lemma \ref{density_dimension} we find a nonempty Zariski open subset $\CO$ of  $(\A^{n+1})^m$ such that for any sequence $\vec{f}=f_1,\ldots,f_m$ of degree one polynomials of $\CO$ the sequence $\vec{f}$ is secant for $\frak{a}$ such that $\dim_K A/(\vec{f})$ is constant and  $\dim_{K(\vec{f})} S^{-1}(\vec{f})A=\sum_{j=1}^q \deg V(\frak{p}_j)\, \ell_j$ holds.

From Corollary \ref{dense_local} we conclude that there exists a sequence $\vec{f_0}=f_1^0,\ldots,f_m^0$ of degree one polynomials of $\CO$ satisfying the equality
\[
\dim_{K(\vec{f_0})} S^{-1}(\vec{f_0})A=\dim_K A/(\vec{f_0}).
\]
This implies
\[\dim_K A/(\vec{f_0})=\sum_{j=1}^q \deg V(\frak{p}_j)\, \ell_j.\]
Hence, for an arbitrary sequence $\vec{f}=f_1,\ldots,f_m$ of degree one polynomials belonging to $\CO$ the sequence $\vec{f}$ is secant for $\frak{a}$ and it holds
\[\dim_K A/(\vec{f})=\dim_K A/(\vec{f_0})=\sum_{j=1}^q \deg V(\frak{p}_j)\, \ell_j.\]
\end{proof}

We may simplify the somewhat complicated formulation  of Theorem \ref{degree_formula} saying that a generic sequence $\vec{f}=f_1,\ldots,f_m$ of degree one polynomials is secant for $\frak{a}$ and $\dim_K A/(\vec{f})=\sum_{j=1}^q \deg V(\frak{p}_j)\, \ell_j$ holds. In this sense the word generic refers always to the existence of a nonempty Zariski open set which not always is made explicit.

Using this terminology we may define the degree of the equidimensional ideal $\frak{a}$ in two different ways as follows.

\begin{definition} \label{degree-primary}
The \emph{degree $deg(\frak{a})$ of the equidimensional ideal $\frak{a}$} may be equivalently defined as
\begin{enumerate}
\item[$i)$] $\deg(\frak{a}):=\dim_K A/(\vec{f})$ for a generic sequence $\vec{f}=f_1,\ldots,f_m$ of degree one polynomials of $K[\vec{X}]$

    or
\item[$ii)$] $\deg(\frak{a}):=\sum_{j=1}^q \deg V(\frak{p}_j)\, \ell_j$, where $\frak{p}_1,\ldots,\frak{p}_q$ are the isolated primes of $\frak{a}$ and $\ell_1,\ldots,\ell_q$ are the lengths of the Artinian local rings $A_{\frak{p}_1},\ldots,A_{\frak{p}_q}$.
\end{enumerate}
The formulation $ii)$ for the degree of $\frak{a}$ was introduced in \cite{bayer} for homogeneous ideals whereas the formulation $i)$ seems new for non-homogeneous ideals and represents our ``workable" notion of degree (see Section \ref{computing}).
\end{definition}

The next statement is a straightforward consequence of this definition:

\begin{proposition} \label{subset}
Let $\frak{a}$ and $\frak{b}$ two equidimensional ideals in the ring $K[\vec{X}]$ of the same dimension $m$. If $\frak{a}\subseteq \frak{b}$ then $\deg(\frak{b})\le \deg(\frak{a})$.
\end{proposition}

\begin{proof}
Since the ideals have the same dimension we can take the same generic degree one polynomials $\vec{f}$ for both ideals. Thus $\frak{a}+(\vec{f})\subseteq \frak{b}+(\vec{f})$ and the proposition follows.
\end{proof}\\

If the ideal $\frak{a}$ is generated by a single polynomial its degree agrees with the total degree of the polynomial which generates it.

\begin{proposition} \label{single}
Let $\frak{a}$ be the ideal generated by a non constant polynomial $g$. The $\deg(\frak{a})=\deg(g)$.
\end{proposition}

\begin{proof}
It suffices to observe that after a generic linear change of coordinates, the degree of the polynomial $g$ does not change if $n-1$ variables are generically specialized and it agrees with the total degree of $g$.
\end{proof}\\

In the case of general polynomial ideals, as customary, we extend our notion of degree as follows.

\begin{definition} \label{degree-any}
Let $I\subset K[\vec{X}]$ be an arbitrary proper polynomial ideal with isolated primary components $\frak{q}_1,\ldots,\frak{q}_t$. We define:
\[
\deg(I):= \sum_{h=1}^t \deg(\frak{q}_h).
\]
\end{definition}

\begin{remark} \label{radical}
For radical polynomial ideals Definitions \ref{degree-primary} and \ref{degree-any} coincide with the usual notions of geometric degree of (equidimensional or arbitrary) algebraic closed subvarieties of affine spaces following \cite{Heintz83}.
\end{remark}

\subsection {On the B\'ezout Inequality}

In view of the B\'ezout Inequality \cite{Heintz83,fulton,vogel} for affine varieties one might expect that for arbitrary ideals $I,J\subset K[\vec{X}]$ the following estimation holds:
\[\deg(I+J)\le \deg(I).\deg(J).\]

That this may become wrong shows the following example.

\begin{example}
\rm{Consider the one-dimensional ideal $I=(X_1^3,X_1^2X_2)\subset K[X_1,X_2]$ whose primary decomposition is $I=(X_1^2)\cap (X_1^3,X_2)$. The primary isolated component is $(X_1^2)$ while $(X_1^3,X_2)$ is the embedded one. By Definition \ref{degree-primary} we have $\deg(I)=\deg((X_1^2))=2$, since for any generic linear polynomial $f_1:=aX_1+bX_2+c$ the ring $K[\vec{X}]/(X_1^2, f_1)$ has the $K$-basis $\{1,X_1\}$.}

\noindent \rm{Take the ideal $J:=(X_2)\subset K[\vec{X}]$, which has degree one, and consider the degree of the sum $I+ J$. We have \[I+ J=(X_1^3,X_1^2X_2,X_2)=(X_1^3,X_2),\]
which is a $0$-dimensional $(X_1,X_2)$-primary ideal. Clearly $\deg(I+J)=3$.}

\noindent \rm{But on the other hand we have $\deg(I).\deg(J)=2.1=2<3$.}
\end{example}

The following example illustrates that there is no chance to obtain an intrinsic B\'ezout Inequality (which depends only on the degrees of the ideals but not their generators).

\begin{example}
\rm{Let $k\in \N$ and $I=(X_1^k, X_1X_2)\subset K[X_1,X_2]$. It is easy to see that the primary decomposition of $I$ is: $I=(X_1)\cap (X_1^k,X_2)$ and then $\deg (I)=\deg((X_1))=1$. By adding the ideal $(X_2)$ we have}
\[
\deg(I+(X_2))=\deg(X_1^k,X_1X_2,X_2)=\deg(X_1^k,X_2)=\]
\[={\textrm{length}}\, (K[X_1,X_2]/(X_1^k,X_2))=k.
\]
\rm{Observe that the degree of the sum of the ideals depends on the degree of the generators of $I$ but not on the degree of $I$.}
\end{example}

Nevertheless, Proposition \ref{prop_bezout1} implies the following B\'ezout-type Inequality for equidimensional ideals.

\begin{theorem} \label{bezout_suc_regular}
Let $K$ be of characteristic zero and let $\frak{a}\subset K[\vec{X}]$ be an equidimensional ideal of dimension $m>0$. Let $f_1,\ldots,f_k$ be a regular sequence, not necessarily maximal, for $\frak{a}$. Then the inequality
\[ \deg(\frak{a}+(f_1,\ldots,f_k))\le \deg(\frak{a})\prod_{i=1}^k \deg(f_i)\]
holds.
\end{theorem}

\begin{proof}
Since $\frak{a}$ is assumed equidimensional and $f_1,\ldots,f_k$ is  a regular sequence, Krull's Principal Ideal Theorem implies that ideal $\frak{b}:=\frak{a}+(f_1,\ldots,f_k)$ is also equidimensional. Combining Proposition \ref{density:corol2} with Theorem \ref{degree_formula} we see that there exists a regular sequence $f_{k+1},\ldots,f_m$  of degree one polynomials such that $\deg(\frak{b})=\dim_K K[\vec{X}]/\frak{b}+(f_{k+1},\ldots,f_m)$ holds.

Let $S:=K[f_1,\ldots,f_k,f_{k+1},\ldots,f_m]\setminus \{0\}$ and $L:=K(f_1,\ldots,f_k,f_{k+1},\ldots,f_m)$, then Lemma  \ref{inequality} states the inequality
\[
\deg(\frak{b})=\dim_K A/(f_1,\ldots,f_k,f_{k+1},\ldots,f_m)\le \dim_L S^{-1}(A).
\]
On the other hand, taking into account that the characteristic of $K$ is zero, from Proposition \ref{prop_bezout1} $i)$ we deduce
\[
\dim_L S^{-1}(A)\le \left(\sum_{j=1}^q \deg(V({\frak p}_j)) \, \ell_j\right) \prod_{i=1}^m \deg(f_i)= \deg(\frak{a})\prod_{i=1}^k \deg(f_i).
\]
This implies the theorem.
\end{proof}

\section{A Masser-W\"ustholz type degree bound for non-homogeneous  polynomial ideals}

The constructions in this section are inspired by \cite{MW}.

Let $\frak{a}$ be an arbitrary non-zero proper ideal of the polynomial ring $R:=K[\vec{X}]$ where $K$ is an algebraically closed field of characteristic zero and $\vec{X}:=(X_1,\ldots,X_n)$ is a set of variables. Denote by $r:=\textrm{ht}(\frak{a})$ the height of the ideal $\frak{a}$.

From a primary decomposition of $\frak{a}$ we obtain a decomposition of $\frak{a}$ as follows:
\[
\frak{a}=\frak{Q}_r\cap \frak{Q}_{r+1}\cap\cdots\cap \frak{Q}_{n}\cap \frak{I},
\]
where, for each $j=r,\ldots,n$ the ideal $\frak{Q}_{j}$ is the intersection of all isolated primary components of $\frak{a}$ having height $j$, or the whole ring $R$ otherwise. The ideal $\frak{I}$ is the intersection of the embedded primary components.

For any $j=r,\ldots,n$ such that $\frak{Q}_j\ne R$, let $\frak{Q}_j=\bigcap_{i=1}^{s_j}\frak{q}_{ji}$ be its primary decomposition. Observe that $\frak{Q}_j$ is unmixed of height $j$.

\subsection{A family of suitable multiplicative sets related to the ideal $\frak{a}$}

In this section we introduce suitable simple multiplicative sets such that the respective localizations detect each equidimensional component of the ideal $\frak{a}$ (see Proposition \ref{local} below).

With the previous notations, for each pair $(k,\ell)$ such that $\frak{Q}_{k}\ne R$ and $1\le \ell\le s_{k}$, there exists a point $z_{k\ell}\in \A^n$ lying in the variety $V(\frak{q}_{k\ell})$ but outside the union of the remaining irreducible components $V(\frak{q}_{ji})$, with $j\ne k$ or $i\ne \ell$ if $j=k$, and the immerse variety $V(\frak{I})$. In purely idealistic terms, there exists a maximal ideal $\frak{m}_{k\ell}$ such that $\frak{q}_{k\ell}\subseteq \frak{m}_{k\ell}$ but  $\frak{q}_{ji}\nsubseteq \frak{m}_{k\ell}$ for all the other pairs $(j,i)$ and $\frak{I}\nsubseteq \frak{m}_{k\ell}$.

For each index $j=r,\ldots,n$ such that $\frak{a}$ has isolated components of height $j$ we introduce the multiplicative set \[ S_{j}:=R\setminus \bigcup _{i=1}^{s_j} \frak{m}_{ji}.\] If there is no isolated component of $\frak{a}$ with height $j$ we define $S_j:=R\setminus\{0\}$.

From $\frak{q}_{k\ell}\subseteq \frak{m}_{k\ell}$ we infer  $\frak{q}_{k\ell}\cap S_{k}=\emptyset$. On the other hand, for $j\ne k$, we have $\frak{q}_{k\ell}\cap S_j\ne \emptyset$. If $S_j=R\setminus\{0\}$ this is obvious. If $S_j\ne R\setminus\{0\}$ the assumption $\frak{q}_{k\ell}\cap S_j=\emptyset$ implies the inclusion $\frak{q}_{k\ell}\subset \bigcup_{i=1}^{s_j}\frak{m}_{ji}$. By \cite[Proposition 1.11]{Atiyah-Macdonald}, there exists then an index $i$ with $\frak{q}_{k\ell}\subset \frak{m}_{ji}$, in contradiction with the choice of the maximal ideals and $j\ne k$. A similar argument shows $\frak{I}\cap S_j\ne \emptyset$ for all $j=r,\ldots n$. Namely, if $\frak{I}$ is disjoint from $S_j$ then $\frak{I}$ must be included in suitable maximal ideal $\frak{m}_{ji}$, which again contradicts the choice of the maximal ideals.

With these considerations we have

\begin{proposition} \label{local}
For any index $k=r,\ldots,n$ the equality
\[ S_{k}^{-1}(\frak{a})=S_{k}^{-1}(\frak{Q}_{k})
\]
holds in the fraction ring $S_k^{-1}R$.

In particular, if $\frak{Q}_k\ne R$, the ideal $S_k^{-1}(\frak{a})$ is unmixed of height $k$ (or equivalently unmixed of dimension $n-k$).
\end{proposition}
\begin{proof}
By the previous arguments we have \[S_k^{-1}(\frak{a})=S_k^{-1}(\frak{Q}_r)\cap S_k^{-1}(\frak{Q}_{r+1})\cap\cdots\cap S_k^{-1}(\frak{Q}_{n})\cap S_k^{-1}(\frak{I})=\]
\[
=\bigcap_{ji} S_k^{-1}(\frak{q}_{ji})\ \cap S_k^{-1}(\frak{I})=\bigcap_{ki} S_k^{-1}(\frak{q}_{ki})=S_k^{-1}(\frak{Q}_k).
\]
Thus, in case $\frak{Q}_k\ne R$, the ideal $S_k^{-1}(\frak{a})$ is unmixed of height $k$ because $\frak{Q}_k$ has this property.
\end{proof}

\subsection{A suitable local regular sequence contained in $\frak{a}$}
\ \medskip

Let $\vec{g}:=g_1,\ldots,g_s$ be a system of generators of $\frak{a}$ with degrees $D_1\ge D_2\ge\cdots\ge D_s$, respectively. Since $\frak{a}$ is assumed generated by $s$ many polynomials, Krull's Principal Ideal Theorem (see \cite[Corollary 11.16]{Atiyah-Macdonald}) implies that in the primary decomposition of $\frak{a}$ only unmixed components $\frak{Q}_k$ with $k\le s$ may appear.

\begin{lemma} \label{regular sequences}
Fix an index $k=r,\ldots,n$ with $\frak{Q}_k\ne R$. Then, there exist polynomials $p_1,\ldots,p_k\in \frak{a}$ such that for all $j$, $1\le j\le k$, the following conditions are satisfied:

\begin{enumerate}
\item[i)] The polynomial $p_j$ is a generic linear combination of the polynomials $g_j,\ldots,g_s$ (in particular,  $\deg(p_j)=\deg(g_j)=D_j$).
\item[ii)] $p_1,\ldots,p_j$ is a regular sequence in the localized ring $S_k^{-1}(R)$.
\item[iii)] If $\frak{a}_j:=(p_1,\ldots,p_j)$ and $\frak{a}_j^*:=S_k^{-1}(\frak{a}_j)\cap R$, the ideal $\frak{a}^*_j$ is an unmixed ideal of height $j$ in $R$.
\item[iv)] The inequality $\deg(\frak{a}^*_{j})\le \deg(\frak{a}^*_{j-1}) \, D_j$ holds.
\end{enumerate}
\end{lemma}

\begin{proof}

The proof runs by induction on $j$, where $1\le j\le k$.

For $j=1$ let $p_1$ be a generic linear combination of the generators $g_1,\ldots,g_s$. Clearly $p_1$ satisfies the conditions $i),ii),iii)$ for $j=1$ because $p_1$ is not invertible in  $S_k^{-1}(R)$ since it belongs to the unmixed $(n-k)$-dimensional ideal $S_k^{-1}(\frak{a})$ (see Proposition \ref{local}). Remark that, even if condition $iv)$ is vacuous for $j=1$, the inequality $\deg(\frak{a}^*_{1})\le \deg(p_1)$ holds (see Propositions \ref{subset} and \ref{single}).

Suppose now that the lemma holds for $1\le j<k$ and let $p_1,\ldots,p_j$ be polynomials verifying conditions $i)-iv)$. Let  \[S_k^{-1}(\frak{a}_j)=\displaystyle{\bigcap_{i=1}^{q} \frak{h}_{i}}\] be a primary decomposition of the ideal $S_k^{-1}(\frak{a}_j)$ in the ring $S_k^{-1}R$. Remark that all the primary components in this decomposition are isolated and $(n-j)$-dimensional because $S_k^{-1}R$ is a Cohen-Macaulay ring and $S_k^{-1}(\frak{a}_j)$ is generated by the regular sequence $p_1,\ldots,p_j$.

For each $i=1,\ldots,q$ consider the linear subspace \[T_{i}:=\{ \mu\in K^{s-j-1}\ |\ \sum_{t\ge j+1} \mu_tg_t\in \sqrt{\frak{h}_{i}}\},\] where $\sqrt{\frak{h}_i}$ denotes the radical of $\frak{h}_i$. If for some $i$, $T_{i}$ is the whole space $K^{s-j-1}$, then $g_{j+1},\ldots,g_s\in\sqrt{\frak{h}_{i}}$. Then, since $p_j\in \frak{h}_i$ and $p_j$ is a generic linear combination of $g_j,g_{j+1},\ldots,g_s$ we infer that $g_j$ also belongs to $\sqrt{\frak{h}_{i}}$ and by repeating this argument we conclude that $S_k^{-1}(\frak{a})\subset \sqrt{\frak{h}_{i}}$. But then, by Proposition \ref{local}, we have the inequality of heights $k\le j$, which contradicts the choice of $j$.

Therefore, any $T_{i}$ is a proper linear subspace of $K^{s-j-1}$, and in particular, a generic vector $\mu\in K^{s-j-1}$ verifies $\mu\notin \bigcup_i T_i$ and so, the associated polynomial $p_{j+1}:=\sum_{t\ge j+1} \mu_tg_t$ is not a zero divisor modulo the ideal $S_k^{-1}(\frak{a}_j)$. Moreover, $p_{j+1}$ is not invertible modulo $S_k^{-1}(\frak{a}_j)$ because of the inclusion $S_k^{-1}((p_{j+1})+ \frak{a}_j)\subseteq S_k^{-1}(\frak{a})$ and Proposition \ref{local}. Hence, $p_1,\ldots,p_{j+1}$ is a regular sequence in $S_k^{-1}R$ and conditions $i)$ and $ii)$ are satisfied for $j+1$.

Condition $iii)$ is a consequence of $ii)$ by Macaulay's Theorem applied to the Cohen-Macaulay ring $S_k^{-1}R$ and the well-known fact that the contraction to $R$ of a primary ideal in $S_k^{-1}R$ remains primary of same dimension.

We finish the proof showing condition $iv)$: Since $p_1,\ldots,p_{j}$ is a regular sequence in $S_k^{-1}R$, the ideals $S_k^{-1}(\frak{a}_j)$ and $\frak{a}^*_j$ are both equidimensional of dimension $n-j$. On the other hand,  as $p_{j+1}$ is regular with respect to $S_k^{-1}(\frak{a}_j)$, then it is also regular with respect to $S_k^{-1}(\frak{a}_j)\cap R=\frak{a}^*_j$. In particular, $\frak{a}^*_j+(p_{j+1})$ is equidimensional of dimension $n-j-1$.

Since $\frak{a}^*_{j}\subseteq \frak{a}^*_{j+1}$ and $p_{j+1}\in \frak{a}^*_{j+1}$ holds, we have the inclusion $\frak{a}^*_{j}+(p_{j+1})\subseteq \frak{a}^*_{j+1}$ and both ideals are equidimensional of dimension $n-j-1$. Thus, Proposition \ref{subset} implies \[\deg(\frak{a}^*_{j+1})\le \deg(\frak{a}^*_{j}+(p_{j+1})).\]

Finally, by Theorem \ref{bezout_suc_regular} applied to the ideal $\frak{a}^*_{j}$ and the polynomial $p_{j+1}$, we have the inequalities \[\deg(\frak{a}^*_{j+1})\le \deg(\frak{a}^*_{j}+(p_{j+1}))\le \deg(\frak{a}^*_{j})\deg(p_{j+1})= \deg(\frak{a}^*_{j})\, D_{j+1}\]

and the lemma is proved.
\end{proof}

\subsection{B\'ezout Inequality in Masser-W\"ustholz style}

Lemma \ref{regular sequences} implies the following B\'ezout Inequality which appears in \cite{MW} in the projective case with the usual notion of degree for homogeneous ideals:

\begin{theorem} \label{weak}
Let $K$ be an algebraically closed field of characteristic zero, $\vec{X}:=(X_1,\ldots,X_n)$ variables over $K$ and $\frak{a}\subset R:=K[\vec{X}]$ a non-zero and proper polynomial ideal generated by polynomials $g_1,\ldots,g_s$ of total degrees  $D_1\ge\cdots\ge D_s$, respectively. Fix an index $k$, $1\le k\le n$, such that in the primary decomposition of $\frak{a}$ there exist isolated components of height $k$ and denote by $\frak{Q}_k$ the intersection of this components. Then:
\begin{enumerate}
\item[i)] The inequality $\deg(\frak{Q}_k)\le D_1\ldots D_k$ holds.
\item[ii)] The inequality
$\deg(\frak{a})\le \sum_{k\in \mathcal{C}} D_1\ldots D_k$,
holds, where $\mathcal{C}$ is the set of those $k$ such that the ideal $\frak{a}$ has isolated primary components of height $k$.
\end{enumerate}
\end{theorem}

\begin{proof}
Following the previous lemma, consider the polynomials $p_1,\ldots,p_k$ defining a regular sequence in the localized ring $S_k^{-1}(R)$. Since $ \frak{a}_k:=(p_1,\ldots,p_k)\subseteq \frak{a}$, we have $S_k^{-1}(\frak{a}_k)\subseteq S_k^{-1}(\frak{a})=S^{-1}_k(\frak{Q}_k)$ (Proposition \ref{local}). Hence, $\frak{a}^*_k:=S_k^{-1}(\frak{a}_k)\cap R\subseteq \frak{Q}_k$, and both ideals are unmixed of height $k$ (and in particular, equidimensional of dimension $n-k$). From Proposition \ref{subset}  we conclude $\deg(\frak{Q}_k)\le \deg (\frak{a}^*_k)$.

By iteration of Condition $iv)$ in Lemma \ref{regular sequences} we obtain $\deg (\frak{a}^*_k)\le D_1\ldots D_k$ which implies inequality $i)$.

The assertion $ii)$ is immediate from $i)$ and the definition of degree of arbitrary ideals (Definition \ref{degree-any}).
\end{proof}

\section{Computing the degree of an equidimensional polynomial ideal} \label{computing}

Let $\frak{a}$ be an equidimensional ideal of the polynomial ring $\Q[\vec{X}]$ in $n$ variables $\vec{X}:=(X_1,\ldots,X_n)$. Assume that the dimension $m$ and a system of generators $g_1,\ldots,g_s\in \Z[\vec{X}]$ of the ideal $\frak{a}$ are known. Let $d$ and $\sigma$ be upper bounds for the degrees of $g_1,\ldots,g_s$ and the bit-sizes of their coefficients.

The goal of this section is to discuss the complexity character of the problem of computing $\deg(\frak{a})$.

Of course, one could compute (uniformly and deterministically) a primary decomposition of the ideal $\frak{a}$ (see \cite{GTZ}) and determine $\deg(\frak{a})$ by means of Definition \ref{degree-primary}  $ii)$. This would involve a computational cost which is doubly exponential in $n$.\\

We present here a probabilistic approach which is more efficient and discuss then whether its complexity can be improved.\\

Let $T_{ij}$, $1\le i\le m$, $0\le j\le n$ be new indeterminates, $\vec{T}=(T_{ij})_{{1\le i\le m}\atop{0\le j\le n}}$ and let for $1\le i\le m$ $\displaystyle{f_i:=\sum_{j=1}^n T_{ij}X_j\, +\, T_{i0}}$. Let $\vec{f}:=f_1,\ldots,f_m$ and fix a monomial order for $\vec{X}$.

Let $\frak{b}$ the ideal generated by $\frak{a}$. From the proof of Lemma \ref{density_dimension} and Theorem \ref{degree_formula} we deduce
\[ \deg(\frak{a})=\dim_{\Q(\vec{T})} \Q(\vec{T})[\vec{X}]/\frak{b}.
\]

Since $\frak{b}$ is an ideal of dimension zero in $\Q(\vec{T})[\vec{X}]$ we have just to compute a Gr\"obner basis of $\frak{b}$ from the generators $g_1,\ldots,g_s$ and $f_1,\ldots,f_m$.

Following \cite[Theorem 3.3]{DFGS} this can be done using $(sd^{n^2})^{O(1)}$ arithmetic operations in $\Q(\vec{T})$.

Applying \cite[Theorem 4.4]{HS} we obtain a non-uniform deterministic or uniform probabilistic algorithm which computes $\deg(\frak{a})$ by means of $(sd^{n^2})^{O(1)}$ arithmetic operations in $\Q$.

The non-uniform deterministic version of the algorithm is based on a hitting sequence of integers having bit-size  $(sd^{n^2})^{O(1)}$. This sequence has to be chosen probabilistically in the uniform complexity model. The whole algorithm requires therefore
$(\sigma sd^{n^2})^{O(1)}$ bit operations.

Putting everything together we obtain the following complexity statement.

\begin{theorem} \label{total_complexity}
There exists a uniform probabilistic algorithm implementable on a Turing Machine with advice which computes $\deg(\frak{a})$ in time $(\sigma sd^{n^2})^{O(1)}$
\end{theorem}

This result raises two questions. What is the \emph{uniform deterministic} complexity of computing $\deg(\frak{a})$? This question seems to be out of reach with the actual techniques.

The other question asks whether $\deg(\frak{a})$ can be computed probabilistically using $(\sigma sd^{n})^{O(1)}$ bit operations.

This question can be answered positively if it is possible to guess probabilistically generic degree one polynomials $f_1,\ldots,f_m$ of $\Z[\vec{X}]$ of bit size $(\sigma sd^{n})^{O(1)}$.

In this case the probabilistic algorithm of Lakshman \cite{lakshman}, techniques of \cite{GH}, a suitable arithmetic B\'ezout Inequality (see e.g. \cite[Th\'eor\`eme 2]{BGS} or \cite[\S 1.2.4]{KPS}) and efficient factorization of univariate polynomials over $\Z$ (see for instance \cite[Corollary 16.25]{gathen}) can be combined to obtain the desired complexity result. We do not go into the (lengthy) details of this approach.

\section*{Conclusion}

We introduced a suitable notion of degree for non-homogeneous polynomial ideals and proved extrinsic B\'ezout Inequalities for this notion. We argued that an intrinsic B\'ezout Inequality for the sum of two ideals is unfeasable. We exhibit a probabilistic algorithm of single exponential complexity which computes the degree of an equidimensional ideal.

\end{document}